%% file: Allerton_Graph.tex
\documentclass[conference]{IEEEtran}
\IEEEoverridecommandlockouts

\usepackage[left=0.75in,right=0.75in,top=1in,bottom=0.75in]{geometry}

\usepackage{ifpdf}
\usepackage{cite}
\usepackage[pdftex]{graphicx}
\usepackage{amsmath}
\usepackage{algorithmic}
\usepackage{array}
\usepackage{subfigure}
\usepackage{caption}
\usepackage[caption=false,font=footnotesize]{subfig}
\usepackage{fixltx2e}
\usepackage{stfloats}
\usepackage{url}
\usepackage{amsthm}
\usepackage{amsmath}
\usepackage{amsfonts}
\usepackage{amssymb}
\usepackage[T1]{fontenc} 
\usepackage{amsmath}
\usepackage[cmintegrals]{newtxmath}
\usepackage{bm} 
\usepackage[all]{xy}
\usepackage{enumitem}
\usepackage{float}
\usepackage{amsmath}
\usepackage[usenames, dvipsnames]{color}
\usepackage{ifpdf}
\usepackage{cite}
\usepackage[pdftex]{graphicx}
\usepackage{amsmath}
\usepackage{algorithmic}
\usepackage{array}
\usepackage{mathtools}
\DeclarePairedDelimiter{\abs}{\lvert}{\rvert}

\usepackage{caption}
\usepackage[caption=false,font=footnotesize]{subfig}
\usepackage{fixltx2e}
\usepackage{stfloats}
\usepackage{url}
\usepackage{mathtools}
\usepackage{amsthm}
\usepackage{amsmath}
\usepackage{amsfonts}
\usepackage{amssymb}
\usepackage[T1]{fontenc} 
\usepackage{amsmath}
\usepackage[cmintegrals]{newtxmath}
\usepackage{bm} 
\usepackage[all]{xy}
\usepackage{enumitem}
\usepackage{float}
\usepackage{amsmath}
\usepackage[dvipsnames]{xcolor}
\usepackage{multirow}
\usepackage{subfig}
\usepackage{multicol}
\usepackage{graphicx}
\usepackage[caption=false]{subfig}
\usepackage{amssymb}
\usepackage{amsmath}
\usepackage{commath}
\usepackage{graphicx,bm}
\usepackage{verbatim}

\usepackage{diagbox}

\theoremstyle{definition}

\newtheorem{thm}{Theorem}

\newtheorem{example}{Example}

\newtheorem{define}{Definition}

\newcommand{\no}{\nonumber}

\begin{document}
	\title{Improving Privacy in Graphs Through Node Addition}
	\author{\IEEEauthorblockN{Nazanin Takbiri}
		\IEEEauthorblockA{Electrical and\\Computer Engineering\\
			UMass-Amherst\\
			ntakbiri@umass.edu}
		\and
		\IEEEauthorblockN{Xiaozhe Shao}
		\IEEEauthorblockA{Electrical and\\Computer Engineering\\
			UMass-Amherst\\
			xiaozheshao@engin.umass.edu}
		\and
		\IEEEauthorblockN{Lixin Gao}
		\IEEEauthorblockA{Electrical and\\Computer Engineering\\
		UMass-Amherst\\
			lgao@engin.umass.edu}
		\and
		\IEEEauthorblockN{Hossein Pishro-Nik}
		\IEEEauthorblockA{Electrical and\\Computer Engineering\\
			UMass-Amherst\\
			pishro@ecs.umass.edu\thanks{This work was supported by National Science Foundation under grants CCF--1421957, CNS-1525836, CNS--1739462, CNS-1815412, and CCF--1918187.}}
	}

	\maketitle

\begin{abstract}

The rapid growth of computer systems which generate graph data necessitates employing privacy-preserving mechanisms to protect users' identity. Since structure-based de-anonymization attacks can reveal users' identity's even when the graph is simply anonymized by employing na\"ive ID removal, recently, $k-$anonymity is proposed to secure users' privacy against the structure-based attack. Most of the work ensured graph privacy using fake edges, however, in some applications, edge addition or deletion might cause a significant change to the key property of the graph. Motivated by this fact, in this paper, we introduce a novel method which ensures privacy by adding fake nodes to the graph. 

First, we present a novel model which provides $k-$anonymity against one of the strongest attacks: seed-based attack. In this attack, the adversary knows the partial mapping between the main graph and the graph which is generated using the privacy-preserving mechanisms. We show that even if the adversary knows the mapping of all of the nodes except one, the last node can still have $k-$anonymity privacy. 

Then, we turn our attention to the privacy of the graphs generated by inter-domain routing against degree attacks in which the degree sequence of the graph is known to the adversary. To ensure the privacy of networks against this attack, we propose a novel method which tries to add fake nodes in a way that the degree of all nodes have the same expected value.

\end{abstract}

\begin{IEEEkeywords}
Graph data, Autonomous System (AS)-level graph, Inter-domain routing, Privacy-Preserving Mechanism (PPM), anonymization and de-anonymization, structural attack, Seed-based attack, $k-$anonymity, $k-$automorohism, $k-$isomorphism.
\end{IEEEkeywords}


\input{introduction}

\input{framework}

\input{conclusion}

\appendices

\bibliographystyle{IEEEtran}
\bibliography{REF}
%
%
%
%

\end{document}

%% file: introduction.tex
\section{Introduction}
\label{into}

Nowadays, a huge amount of data is generated from various different computer systems which can be modeled by a graph data. Social network data~\cite{Backstrom2006GroupFI,Kossinets2008TheSO,bergami2019social,kumar2010structure}, communication data~\cite{Lakshmanan2005ToDO}, Internet peer-to-peer networks and other network topologies~\cite{Ripeanu2002MappingTG,asharov2017privacy}, mobility traced-based contact data ~\cite{Srivatsa2012DeanonymizingMT} are some examples of computer systems and services which generate graph data. In these graphs, nodes represent users/systems and edges represent relationship between users/systems~\cite{wasserman1994social}.

The necessity of sharing graph data for research purposes, data mining task, and commercial applications~\cite{Ji2017GraphDA} presents a significant privacy threat to the users/systems~\cite{narayanan2009anonymizing} --- even when the graph is simply anonymized --- since the adversary can leverage their side-information about the structural graph to infer the private information of the users/systems which generated the graph~\cite{backstrom2007wherefore, Ji2014StructuralDD, Ji2014StructureBD,Ji2015OnYS,narayanan2009anonymizing}.

The structure-based attacks have been introduced to graph data by~\cite{narayanan2009anonymizing, backstrom2007wherefore}. The structure-based attacks are aimed to de-anonymize anonymized users in terms of their uniquely distinguishable structural characteristics. There are different kinds of structure-based attacks that can be mainly categorized in the following groups:
\begin{itemize}
	\item \textbf{Degree Attacks:} Assume the adversary knows the degree sequence of the graph, thus the adversary can use the degree sequence of the graph to uniquely identify one user/system if its degree is unique~\cite{pedarsani2011privacy,yartseva2013performance}. $k-$degree anonymity is proposed by~\cite{liu2008towards} in order to protect graph against this specific attack in a way that for each node, there exist at least $k-1$ other nodes with the same degree. 
	\item \textbf{$1$-Neighborhood Attacks:} Assume the adversary knows the immediate neighbors of the target users, so they have complete information about the nodes adjacent to the target node. $k-$neighborhood anonymity is proposed by~\cite{zhou2008preserving} to protect graph against the adversary who has knowledge about the neighborhood of the target node. In this privacy mechanism, for each node, there exist at least $k-1$ other users who have same neighborhood.~\cite{jin2011preserving,wu2013d} also extend the previous work, to design an algorithm to defend against the adversary who has knowledge about the $d-$neighborhood of a target node.
	\item \textbf{Sub-graph Attacks:} The sub-graph attack in which the adversary knows a sub-graph around the target node is the general case of $1-$neighborhood attack.~\cite{hay2008resisting,hay2007anonymizing, backstrom2007wherefore} proposed a method to protect users' identity from this specific kind of attack.
	\item \textbf{Hub-Fingerprint Attacks:} Assume there exist some hubs with high degree and high betweenness centrality~\cite{newman2003structure} in the graph which have been identified in the released network. Now, assume there exists an adversary who has knowledge about distance between these sets of designated hub nodes and a target node, thus they can use this knowledge to break the privacy of the target node.	
\end{itemize}

Zou et al.~\cite{zou2009k} proposed the concept of $k-$automorphism in a way that can provide privacy against all of the above-mentioned structure-based attacks. Otherwise stated, a graph is $k-$automorphic, if for any node in the graph, there exist $k-1$ symmetric nodes based on any structural information such as their degree, their neighborhood, their distance from hubs, etc. The utility of this method characterized by using the number of faked edges added to the graph, thus, this method is specifically useful when the networks have symmetry properties~\cite{lauri2016topics,wang2009symmetry,ying2008randomizing}. However, the privacy mechanism proposed by~\cite{zou2009k} didn't address the privacy of the sensitive relationship between nodes, and in other words, the graph can suffer from path length leakage and edge leakage.~\cite{machanavajjhala2006diversity, wong2006alpha} showed the path length leakage and edge leakage exist even when a graph is $k-$anonymous and $k-$automorphic. Cheng et al.~\cite{cheng2010k} proposed a new $k-$isomorphism privacy preserving mechanism which preserves the privacy of not only nodes but also edges.~\cite{cheng2010k} convert the graph into $k$ disjoint isomorphic sub-graphs and proved each node and each edges can be identified with the probability of $\frac{1}{k}$, thus they satisfy $k-$anonymity. However, the privacy mechanism proposed by~\cite{cheng2010k} decreases the utility of the released graph since the edges between sub-graphs are deleted. In order to compensate this utility loss, Yang et al.~\cite{yang2014secure} proposed a graph anonymization method in which the anonymous graph should satisfy $AK-$secure privacy preserving mechanism to minimize utility loss.

Today, knowledge of the adversary is not just limited to the structural of the graph, but also richer side-information in the form of seeds is available to them; otherwise stated, the adversary knows the mapping between the original graph and the anonymized graph for a subset nodes~\cite{onaran2016optimal,kazemi2015growing,lyzinski2014seeded,kazemi2015can,cullina2016improved,Korula2014AnER,shirani2017seeded,yartseva2013performance}. Access of adversary to this side-information, which is difficult to control, make the previous anonymization technique more vulnerable.~\cite{onaran2016optimal,kazemi2015growing,lyzinski2014seeded,kazemi2015can,cullina2016improved,Korula2014AnER,shirani2017seeded,yartseva2013performance} assume the network graph is generated using Erd\"{o}s-R\'{e}nyi random graph model~\cite{kazemi2016network}, which is not a realistic assumption.~\cite{Ji2015OnYS,Ji2016SeedBasedDQ}, proposed the first theoretical quantification of the perfect de-anonymization of a general setting in which the graph can be generated using any random model.

The bulk of previous work ensured privacy by deleting or adding fake edges. In this paper, we turn our attention to the case that privacy is guaranteed by adding fake nodes. 
Adding fake nodes could be a promising scheme to defend the graph privacy against the adversaries with side-information. Especially, it is necessary to add fake nodes in cases where other methods, such as node deletion, edge deletion and fake edge addition, might significantly decrease the utility of the released graph. For example, to study the inter-domain routing in the Internet, the Internet topology is usually modeled as an Autonomous System-level (AS-level) graph, where each node represents an Autonomous System (AS) which is a network operated by an institution and an edge between two nodes represents that two networks are directly connected~\cite{974527,TimGriffin_SIGCOMM99,Sobrinho:2003:NRP:863955.863963,Sosnovich:2015:AIR:3089605.3089614,complexrelationship,Wang:2009:NBM:1555349.1555375}. In this scenario, the reliability property of a network to the rest of the Internet and the best path from one network to another are essential for the study of the inter-domain routing~\cite{7218436,asharov2017privacy}. Node or edge deletion might change the best path from one network to another or even make some networks unreachable from the rest of networks. Similarly, adding a fake edge between two real networks also changes the reliability property, since the fake edge leads to an additional path between two real networks. 

In this paper, we proposed a novel method called $k-$fold replication method which preserves privacy of systems/users against seed-based attacks through adding fake nodes to the graph. Then, we turn our attention to the privacy of the graphs generated by inter-domain routing against degree attacks in which the degree sequence of the graph is known to the adversary. To address this issue, we propose a novel method which tries to add fake nodes in a way that the degree of all nodes have the same expected values.

The rest of the paper is organized as follows. In Section~\ref{sec:general}, we present the general setting for privacy on graphs: system model, metrics, and definitions. Then, the conditions for achieving privacy by adding fake nodes in the case of seeded-based attack is discussed in Section~\ref{sec:fake}. In Section~\ref{sec:routers}, we discuss how to ensure privacy for the networks using node addition, and in Section~\ref{sec:conclusion}, we conclude from the results.

%% file: framework.tex

\section{A General Setting for Privacy on Graphs}
\label{sec:general}

In a general setting, we are given a graph $G=G(V,E)$, and without loss of generality, we write $V=\{1,2,\cdots n\}$. Given a privacy mechanism $\mathcal{M}$ which is employed to guarantee privacy, a new graph $G^p(V^p,E^p)$ has been produced. In the simplest case, we might have $G^p \simeq G$ or in more advanced settings, we could construct $G^p$ from $G$ by adding fake vertices or adding/deleting edges.

The way $G^p$ is constructed from $G$ depends on the privacy mechanism $\mathcal{M}$ which is designed for the specific context, the constraints and requirements of the problem scenario. There exists a mapping function, $\sigma:V \mapsto V^p$ which is a function that determines the mapping between vertices of $G$ and $G^p$. More specifically, for each vertex $v \in V$, $\sigma(v) \in V^p$ is the corresponding vertex in $G^p$. The adversary tries to identify users' identities by finding the mapping function $\sigma$.

Although we assume the privacy mechanism is known to adversary; the construction of $G^p$ normally involves a randomized component and this randomness is what ensures the privacy. Here, we are interested in guaranteeing a privacy level, and our goal is employing a privacy preserving method to perturb the original graph structure to protect users' privacy while preserving as much data utility as possible.

%
%
%

We first briefly review the terminology that we use in this paper. Note that we adopt the definitions of $k-$automorphism and $k-$isomorphism from~\cite{de2003large,cheng2010k,beineke2004topics,zou2009k}, respectively.
\begin{define}
	\textit{Graph Isomorphism}~\cite{de2003large}: Given two graphs $G=(V^G, E^G)$ and $Q=(V^Q,E^Q)$, graph $Q$ is isomorphic to graph $G$ if there exists a permutation function $(\Pi:V^Q \mapsto V^G)$ such that $(u,v)$ is in the set of graph edges $E^G$ iff $(\Pi(u),\Pi(v))$ is in the set of graph edges $E^Q$.
\end{define}
\begin{define}
	\textit{$k-$Isomorphism}~\cite{cheng2010k}: A graph $G(V^G,E^G)$ is $k-$isomorphic
	if graph $G$ consists of $k$ disjoint subgraphs, i.e., $G = {G_1 \cup G_2 \cup \cdots, \cup G_k}$, where $G_i$ and $G_j$ are isomorphic for $i \neq j$.
\end{define}

\begin{define}
	\textit{Graph Automorphism}~\cite{beineke2004topics} :Given graph $G=(V^G, E^G)$, it is a graph automorphism from graph $G$ to itself if there exists an automorphic function $(\Pi:V^G \mapsto V^G)$ such that $(u,v)$ is in the set of graph edges $E^G$ iff $(\Pi(u),\Pi(v))$ is in the set of graph edges $E^G$.
\end{define}
\begin{define}
	\textit{$k-$Automporphism}~\cite{zou2009k}: A graph $G(V,E)$ is $k-$automorphic if for any node $v$ in the graph, there exist $k-1$ different automorphic functions.
\end{define}

Here, we assume a strong adversary which employs seed-based attack which is defined as:
\begin{define}
	\textit{Seed-Based Attack:} In this attack, we assume the adversary knows: (1) The main Graph ($G$) (or part of it), (2) The graph which is generated by privacy mechanism $\mathcal{M}$ ($G^p$), and (3) Partial of mapping function ($\sigma$), more specifically, the adversary knows the values of $\sigma$ for a subset $V^s \subset V$. The goal of the adversary is to determine the values of $\sigma$ for some vertices in $V-V^s$.
\end{define}

Now, before our method is discussed in detail, the measures of privacy cost and degree of anonymity that we employ are proposed.
\begin{define}
	\textit{Privacy Cost:} The cost of a privacy mechanism is usually formulated as the distance measure between $G$ and $G^p$. In other words, the more we distort $G$ to make $G^p$, the more privacy cost we incur.
\end{define}

Privacy metrics could also depend on the situation. One type of privacy metric can be defined in terms of the minimum number of vertices that are needed to be revealed to the adversary ($\abs{V^s}$), so that the adversary can recover the values of $\sigma$ for some vertices in $V-V^s$. More specifically, we can have the following definition.

\begin{define}\textit{Privacy Tolerance of Node $v$:}
	A privacy mechanism has a \emph{privacy tolerance} $\tau_v$ for a vertex $v \in V$, if the adversary is unable to recover $\sigma(v)$ unless $\abs{V^s}>\tau_v$. The largest value of $\tau_v$ that satisfies this property, is said to be the maximum tolerance of the mechanism for vertex $v$ and we write
	\[ \mathcal{T}_{v}(\mathcal{M})=\tau_v.\]
\end{define}
The value $\mathcal{T}_v(\mathcal{M})$ is specific to a vertex and depends on the structure of $G$, so we can provide the following measure for privacy of the mechanism that does not depend on the graph $G$.
\begin{define}\textit{Privacy Tolerance:}
 A privacy mechanism has a \emph{privacy tolerance} $\tau$ if for all graphs $G$ with $|V|=n$, and all the vertices $v \in V^s$, we have $\mathcal{T}_{v}(\mathcal{M}) \geq \tau$. The largest value of $\tau$ that satisfies this property, is said to be the maximum privacy tolerance of the mechanism and we write
 \[ \mathcal{T}(\mathcal{M})=\tau.\]
\end{define}
Clearly, we have $\mathcal{T}(\mathcal{M}) \leq n-1$. Maximum tolerance gives some measure of privacy, but it does not provide all the needed information. More specifically, it does not give us a measure of the uncertainty of the adversary when $\abs{V^s} < \mathcal{T}_n(\mathcal{M})$.

Now, we define \emph{privacy function} that provides a much more complete picture about the privacy level of a mechanism. As we will see, the maximum tolerance defined above can be easily extracted from the privacy function. Intuitively, the privacy function, $h_n(\Lambda)$, gives us the guaranteed uncertainty about $\sigma(v)$, when the adversary has labels of $\Lambda$ vertices (i.e., $\abs{V^s}=\Lambda$).

We now provide the formal definition of privacy function for a privacy mechanism $\mathcal{M}$. Entropy and mutual information usually provide an effective tools for defining privacy measures~\cite{Nazanin_IT, tifs2016}. If $\abs{V^s}=\Lambda$, we write
\[ V^s=\{u_1,u_2, \cdots, u_\Lambda\}. \]

\begin{define}
	\label{def9}
 \textit{Privacy Function:} For a privacy mechanism $\mathcal{M}$, privacy function which is denoted as $h_n: \{0,1,2, \cdots, n\} \mapsto \mathbb{R}^{+}$ is defined as follows:
 \begin{align}
 \no h_n(\Lambda)&=\min \bigg{\{}H \big{(}\sigma(v)\big| \sigma(u_1),\sigma(u_2), \cdots, \sigma(u_\Lambda) \big{)} :\\
\no &\hspace{1.3 in}|V|=n, |V^s|=\Lambda, v \in V-V^s \bigg{\}},\ \
 \end{align}
 where $H(\cdot | \cdot)$ denotes the conditional entropy. Note that~\cite{diaz2002towards,serjantov2002towards,nilizadeh2014community,ji2018quantifying,ji2016evaluating} also employed entropy to define degree of anonymity achieved by the users of a system towards particular attackers.
\end{define}Note that the minimum is taken over all graphs $G$ with $n$ vertices, all $v \in V$, and all $V^s \subset V$. With this definition, it easy to see
 \[ \mathcal{T}_n(\mathcal{M})= \min \{\tau: h_n(\tau)=0 \}.\]

\section{Achieving Privacy Against Seed-Based Attacks}
\label{sec:fake}

When the Internet is modeled as an AS-level graph, to preserve the utility of the AS-level graph, the anonymization scheme should maintain the key properties, such as the reachability and reliability between networks.
In the studies of the inter-domain routing, it could be essential to figure out the reliability and the best path from one network to another~\cite{7218436,asharov2017privacy}. For example, to derive the global routing table, the best route of each network should be derived based on the AS-level graph~\cite{asharov2017privacy}.

In an AS-level graph, adding fake nodes into the AS-level graph introduces additional fake networks. After that, fake edges can be added between two fake nodes or between a fake node and a real node. These fake edges represent fake connections between two networks. Note, we do not add fake
edges between two origin nodes, since it decreases the utility of the graph.

Adding fake nodes can preserve the utility of an AS-level graph, since the key properties of the graph can be maintained. On one hand, the additional fake nodes do not remove the original paths between two networks in the graph. Therefore, if a path from a real network to another real network exists in the original graph, the path always exists in the generated graph. On the other hand, even if the additional fake nodes and edges might lead to additional paths from a network to another, the the export policy of fake nodes can not altered to guarantee that these additional paths between two real networks are invalid paths in terms of inter-domain routing. That is, a fake network will not announce a route that goes through real networks to its neighboring nodes that represent real networks. Therefore, a path with any fake edge will not be a valid route from one real network to another real network.

In this section, we propose approaches to make a graph $k-$anonymized through adding fake nodes. Here, the adversary has side-information not only about the structural information of the graph, but also in the form of seeds.
Now, since we focus on privacy mechanisms that add fake nodes to preserve privacy, we use the number of fake nodes to measure the cost of a privacy mechanism as follows.
\begin{define}
\textit{Privacy Cost:} Given an original graph $G$ and the generated graph by privacy mechanism $G^p$, the privacy cost is defined as the number of fake nodes added to the original graph, in other words,
$$Cost\left(G,G^p\right) = |V^p|-|V|$$.
\end{define}

\subsection{The Na\"ive Approach}

In this method, we generate exactly $k-1$ copies of the original graph to achieve a new graph which satisfies $k-$automorphism/$k-$isomorphism.
\begin{figure}[h]
	\centering
	\includegraphics[width = \linewidth]{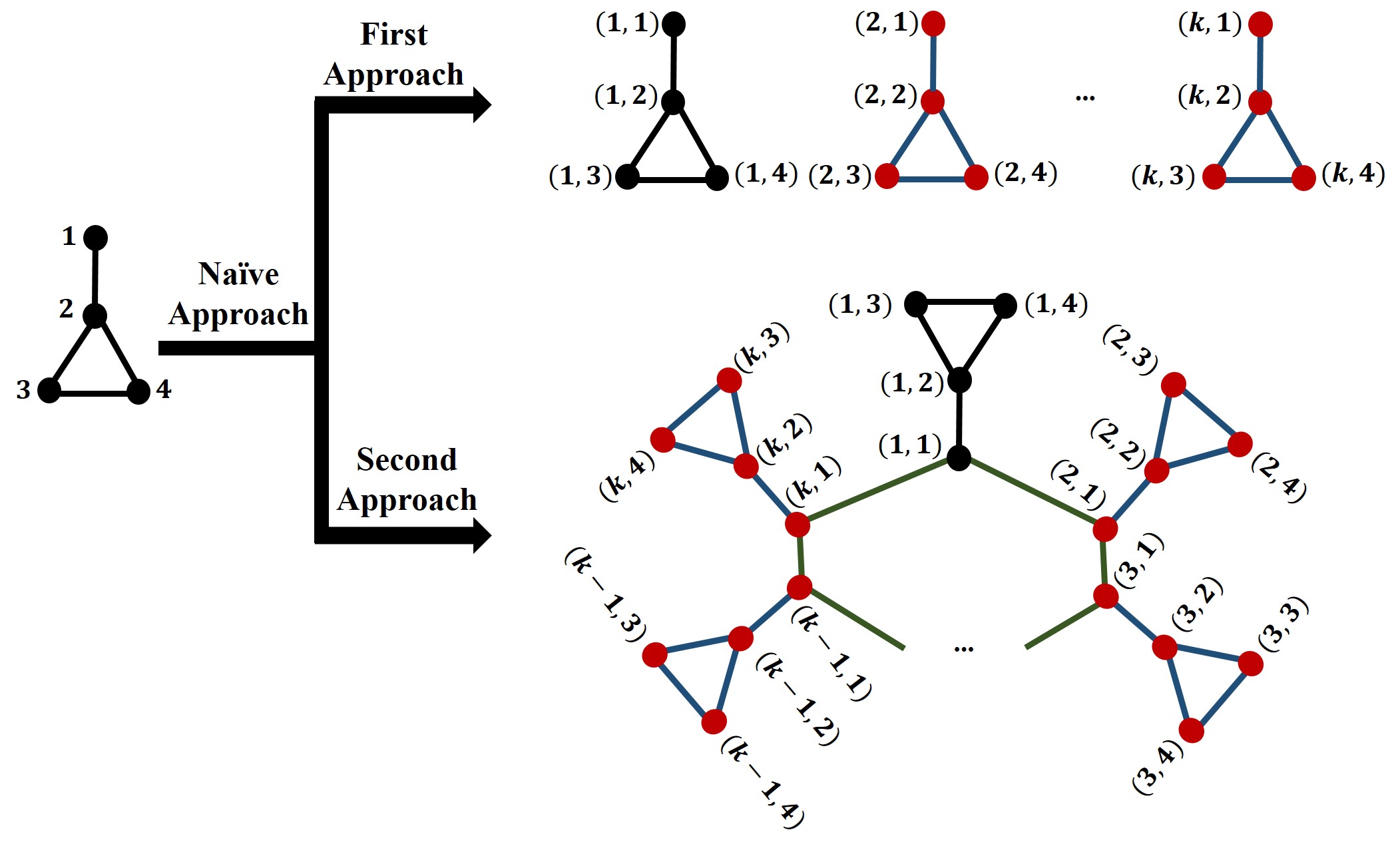}
	\caption{An Example which demonstrates how Na\"ive approach is employed to satisfy $k-$anonymity. First approach satisfies $k-$isomorphism by generation $k-1$ copies of the original graph. Second approach satisfies $k-$automorphism by connecting all of the $k-1$ copied graphs. }
	\label{fig:naive}
\end{figure}

As it is shown in Figure \ref{fig:naive}, the set of nodes of the generated graph can be defined as:
\begin{align}
\no V^p=\left\{(i,j): i \in \{1,2,\cdots, k\}, j \in \{1,2,\cdots, n\}\right\},\ \
\end{align}
thus, the privacy cost can be calculated as:
$$Cost (G, G^p)=(k-1)n.$$
Note that the privacy tolerance of this graph which is generated using Na\"ive approach is one, in other words, if only one of the nodes is known to the adversary, the adversary can recover the whole graph and break privacy of users.
 \[ \mathcal{T}_n(\mathcal{M})=1.\]
Now, by using the fact that the privacy tolerance of this generated graph is equal to one, we can conclude the privacy function of the Na\"ive approach can be calculated as:
\[h_n(\Lambda)=\begin{cases}
\log_{2}k, & \textrm{for } \Lambda=0.\\
0, & \textrm{for } \Lambda \geq 1.
\end{cases}\]
As a result, the "$k-$anonymity" type approaches to privacy are not usually sufficient.

\subsection{The Replication Method}


Motivated by the above discussions, here we introduce a technique to guarantee privacy against seed-based attacks. We call it the replication method, as it can be thought of as replicating the vertices of the original graph in a special way. The basic idea is as follows. Start from any vertex in the graph $v_1 \in G$ and add a new vertex to the graph which is connected to all the neighbors of $v_1$. Call the new graph $G_1$. Now identify another vertex in $v_2 \in G$ and add a new vertex in $G_1$ that is connected to all neighbors of $v_2$ in $G_1$. This will give you $G_2$. Repeat this process until you exhaust all vertices of $G$. At the end you will obtain $G_{n}$ which will be our graph $G^p$. This is a two-fold replication method. You can simply extend this to a $k$-fold replication by repeating the whole process $k-1$ times. Below we formally introduce the technique and show that it provides a high guarantee against seed-based attacks.

The set of nodes of the graph which is generated by $k-$fold replication can be defined as:
\begin{align}
\no V^p&=\big{\{}(i,j): i \in \{1,2,\cdots, k\}, j \in \{1,2,\cdots, n\}, \\
\no &\hspace{1.2 in}\left((i,j),(u,v)\right)\in V^p \text{ iff } (j,v) \in V\big{\}} ,\ \
\end{align}
thus, the privacy cost can be calculated as:
\[Cost(G,G^p)=(k-1)n,\]
which is the same as the Na\"ive method. The number of fake edges needed in the $k-$replicated method can be also calculated as:
$$\Delta |E|=|E^p|-|E|=(k^2-1) |E|.$$

Note that the privacy tolerance of the graph $\mathcal{T}_n(\mathcal{M})$ can be calculated as:
\[ \mathcal{T}_n(\mathcal{M})=n-1, \]
in other words, if the adversary knows the mapping function $(\sigma)$ for $n-1$ nodes--which is the maximum possible value for the privacy tolerance-- they still identify the last node with the probability of $\frac{1}{k}$, so this method can obtain the maximum possible value for the privacy tolerance. 

\begin{thm}
For $k-$fold replication method, the privacy function is bigger than or equal to $\log_2k$ for the case $\Lambda \in\{0, 1, \cdots, n-1\}$. In other words,
$$h_n(\Lambda) = \log_{2} k \text{ for all } \Lambda=0,1,2, \cdots, n-1.$$
\end{thm}
\begin{proof}
     In the first step, we prove that $ h_n(\Lambda) \geq \log_2 k$. To do so, let's assume the adversary knows the mapping function for $n-1$ nodes, and there is only one unknown node. Since the replication method creates a $k$ vertices with the same neighbors, by symmetry, the privacy function can be calculated as
	\begin{align}
	\no h_n(n-1)&=-\sum\limits_{i=1}^{k}p_i\log_2 p_i\\
	\no &=-\sum\limits_{i=1}^{k}\frac{1}{k}log_2 \frac{1}{k}\\
	&=\log_2{k}.\ \
	\label{h1}
	\end{align}
Now, given the fact that conditioning reduces entropy, for all $\Lambda \leq n-1$, we have
\begin{align}
\no &H\left(\sigma(v)|\sigma(u_1), \sigma(u_2),\cdots, \sigma(u_{\Lambda})\right) \geq\\
\no &\hspace{1.5 in}H\left(\sigma(v)|\sigma(u_1), \sigma(u_2),\cdots, \sigma(u_n)\right),\ \
\end{align}
and as a result, for all $\Lambda \leq n-1$,
\begin{align}
\no h_n(\Lambda) &\geq h_n(n-1)\\
&=\log_2k.\ \
	\label{h2}
\end{align}
Now, from (\ref{h1}) and (\ref{h2}), we can conclude the privacy function for the $k-$replication method satisfies
$$h_n(\Lambda) \geq \log_{2} k \text{ for all } \Lambda=0,1,2, \cdots, n-1.$$

In the second step, to show that
$$h_n(\Lambda) \leq \log_{2} k \text{ for all } \Lambda=0,1,2, \cdots, n-1,$$
it suffices to provide examples of scenarios (for all $n$ and $\Lambda=0,1,2, \cdots, n-1$) where
\[H \big{(}\sigma(v)\big| \sigma(u_1),\sigma(u_2), \cdots, \sigma(u_\Lambda) \big{)}\leq \log_{2} k. \]

Consider the graph $G (V,E)$, which is shown in Figure \ref{fig:star}. In this graph, there exist $n-1$ node with degree of one, and one node with degree of $n-1$. Thus, the degree sequence of the original graph is equal to
$$\textbf{d}^G=[\underbrace{1, 1, \cdots,1,}_\text{n-1} n-1],$$

Note that the $k-$fold replication method, increases degree of each node by a factor of $k$, thus, the degree sequence of graph after $k-$fold replication is equal to:
$$\textbf{d}^{G^P}=[\underbrace{k, k, \cdots,k,}_\text{k(n-1)} \underbrace{k(n-1), k(n-1), \cdots,k(n-1)}_\text{k}].$$
This means that in the $k-$fold replication method, the uncertainty of the adversary about the node which has degree of $n-1$ is always less than or equal to $\log_{2}  k$. Therefore,
$$h_n(\Lambda) \leq \log_{2} k \text{ for all } \Lambda=0,1,2, \cdots, n-1.$$

	\begin{figure}[h]
	\centering
	\includegraphics[width = 0.4\linewidth]{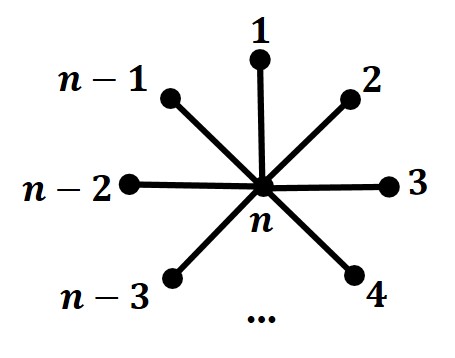}
	\caption{A graph with $n-1$ node with degree of one, and one node with degree of $n-1$.}
	\label{fig:star}
\end{figure}


\end{proof}

The idea behind $k-$fold replication method is explained by the following example:
	\begin{example}
Figure \ref{fig:replication} shows the replication method for the case $k=2$; namely, two-fold replication. In this example, number of nodes ($n$) is equal to $3$, thus the privacy cost is calculated as
$$Cost(G,G^p)=3,$$
 and the privacy tolerance of the graph is equal to
 $$\mathcal{T}_n(\mathcal{M})=2.$$
Thus, the privacy function of this two-replicated method is equal
$h_n(\Lambda) = 1$ for all $\Lambda=0,1,2, \cdots, n-1.$

	\begin{figure}[h]
		\centering
		\includegraphics[width = 0.8\linewidth]{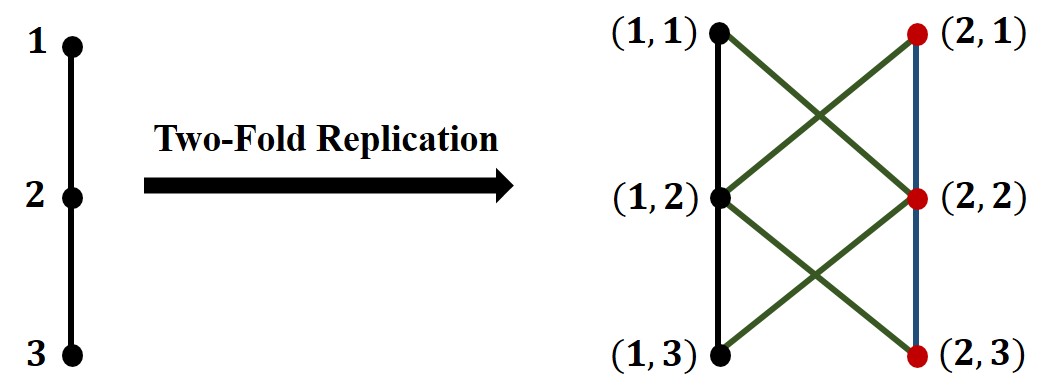}
		\caption{Two-fold Replication Method. Black nodes represent real nodes and red nodes represent added fake nodes. Black edges represent the real edges, green edges represent the fake edges which connect a fake node and a real node, and blue edges represent the fake edges which connect two fake nodes. }
		\label{fig:replication}
	\end{figure}

	\end{example}

\section{Privacy of Networks Through Node Addition}
\label{sec:routers}

In this section, we want to provide privacy regarding the \textit{Degree Attack}. Degree attack is one type of the structure-based attacks, in which, the adversary knows degree sequence of the graph, in other words, they know the degree of all nodes in the graph~\cite{liu2008towards,hay2008resisting}. 

In the scenario of the AS-level graph, the degree of each network can be inferred from publicly available datasets, such as Border Gateway Protocol (BGP) routing information provided by Route Views project~\cite{routeview} and Routing Information Service (RIS) provided by R\'eseaux IP Europ\'eens Network Coordination Center (RIPE NCC)~\cite{ripe}. 
The large networks at the core of the Internet, such as Tier-1 Internet Service Providers (ISPs), have high degree. Through their special degrees, the large networks can be easily mapped to the nodes in the AS-level graph. Then these identified networks can benefit the follow-up de-anonymization greatly, since these identified networks will provide a lot of structural information. When these networks are identified, the privacy information (routing policies) attached to the nodes will be disclosed.

Although the privacy schemes in Section~\ref{sec:fake} can provide privacy regarding various structure-based attacks, it incurs high privacy cost. Namely, $k-1$ replications for each real node are added to the original graph. In this section, we propose a more efficient privacy mechanism in terms of privacy cost against degree attack.

\subsection{A General Setting for Privacy Against Degree Attack}
The original graph such as power law graph is asymmetric in terms of degree distribution~\cite{hay2009accurate}, which reveals a lot of information to the adversary who has knowledge about degree distribution of the graph data.
Assume $G(V, E)$ denotes a graph data with set of nodes $V$ ($|V|=n$) , and set of edges $E$. Without loss of generality, we assume $v \in V =\{1, 2, \cdots, n\}$. Our main goal is protecting identity of all of the nodes of this graph from a strong adversary who has full knowledge about the degree sequence of this graph. Now, assume the adversary knows the $1 \times n$ vector containing the degree of node $v$, $$\textbf{d}=[d_1, d_2, \cdots, d_n],$$where $d_v$ denotes degree of node $v$. In order to preserve the privacy of nodes, a new random graph called $G^P$ is generate in a way that all vertices have the same expected degree.
In order to have the same expected values for all the vertices, we add $m$ fake nodes to the original graph $(G)$. $U=\{1,2, \cdots, m\}$ denotes the set of fake nodes; thus, $u \in \{1,2, \cdots, m\}$. Now, new graph $G^p=(V^p,E^p)$ is constructed using probabilistic method to introduce more uncertainty to the model and as a result, confuse the adversary more. In this method, fake edges which connects two fake nodes or one fake node with one real node randomly; to be more specific,
\begin{itemize}
\item Each edge between a real node and a fake node is included in the graph with probability $p_{vu}$ independent from every other edge, where $v \in \{1,2, \cdots n\}$ and $u \in\{1,2, \cdots,m\}$. There exist $n \times m$ different values for $p_{vu}$'s.

\item Each edge between two fake nodes is included in the graph with probability $q_{uw}$ independent from every other edge, where $u,w \in\{1,2, \cdots,m\}$. There exist ${m \choose 2}$ different values for $q_{uw}$'s.
\end{itemize}

Now, as shown in Figure~\ref{fig:degree1}, for all real nodes $v \in \{1, 2, \cdots, n\}$, the expected value of degree of each real node after this operation can be calculated as:
 \begin{align}
\mathbb{E}[D_v]=d_v+\sum\limits_{u=1}^{m}p_{vu}.\ \
\label{deg_1}
 \end{align}
 	\begin{figure}[h]
 	\centering
 	\includegraphics[width = 1\linewidth]{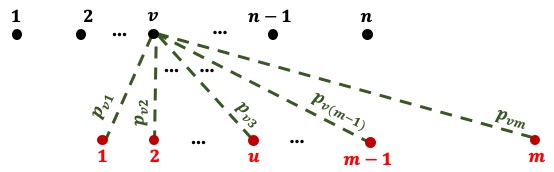}
 	\caption{The real node $v$ and its degree after the anonymization operation is employed to generate the random graph $G^P$. In this figure, for simplicity, the real edges of the graph is not shown.}
 	\label{fig:degree1}
 \end{figure}

 Also, as shown in Figure~\ref{fig:degree2}, for all fake nodes $u \in \{1, 2, \cdots, m\}$, the expected value of degree of each fake node after this operation can be calculated as:
 \begin{align}
\mathbb{E}[D'_u]=\sum\limits_{v=1}^{n}p_{vu}+\sum\limits_{\substack{w=1 \\ w\neq u}}^{m}q_{uw}.\ \
\label{deg_2}
 \end{align}

\begin{figure}[h]
\centering
\includegraphics[width = 1\linewidth]{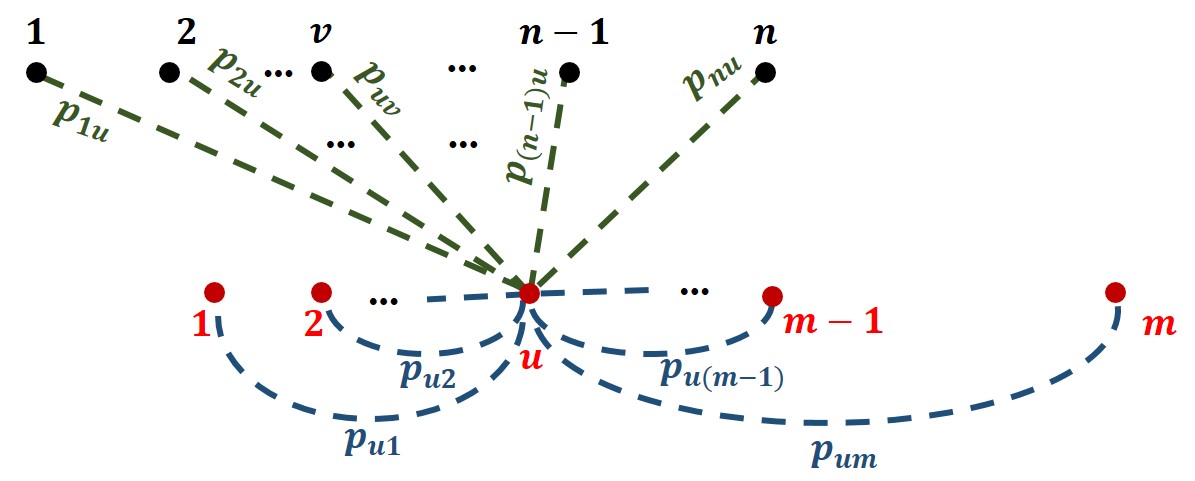}
\caption{The fake node $u$ and its degree after the anonymization operation is employed to generate the random graph $G^P$.}
\label{fig:degree2}
\end{figure}

 Our goal is adjusting the value of $p_{vu}$'s and $q_{uw}$'s in a way that identity of users against degree attack will be protected.

\begin{thm}
	Consider a general graph $G(V,E)$. If all the followings hold
	\begin{enumerate}
		\item $m=n-\frac{2|E|}{a}$, where $a$ is a constant number which is greater than the maximum degree of original graph.
		\item $p_{vu}=\frac{a-d_v}{m}$, for all $v \in \{1,2, \cdots, n\}$ and $u \in \{1,2, \cdots, m\}$.
		\item $q_{vw}=0$, for all $u \in \{1,2, \cdots, n\}$ and $w \in \{1,2, \cdots, m\}$.
	\end{enumerate}
	 then, the expected values of all the real and fake nodes have the same value which is equal to $a$.
\end{thm}
\begin{proof}	
	 After adding $m$ fake nodes which each of them is connected to real nodes with probability of $p_{vu}=\frac{a-d_v}{m}$, for all $v \in \{1, 2, \cdots, n\}$, the expected value of of degree of each real node can be calculated by using (\ref{deg_1}):
	\begin{align}
	\no \mathbb{E}[D_v]&=d_v+m\left(\frac{a-d_v}{m}\right)\\
	&=a,\ \
	\label{real_exp}
	\end{align}
	Also, for all $u \in \{1, 2, \cdots, m \}$, by using (\ref{deg_2}), we have
		\begin{align}
	\no \mathbb{E}[D'_u]&=\sum\limits_{v=1}^{n}p_{vu}+\sum\limits_{\substack{w=1 \\ w\neq u}}^{m}q_{uw}\\
	\no &=\sum\limits_{v=1}^{n}\frac{a-d_v}{m}\\
	\no &=\frac{na-2|E|}{\frac{na-2|E|}{a}}\\
	&=a,\ \
	\label{fake_exp}
	\end{align}
since each fake node is connected to a real node with probability of $p_{vu}=\frac{a-d_v}{m}$, and two fake nodes are connected with probability of $q_{uw}=0$.

Now, from (\ref{real_exp}) and (\ref{fake_exp}), we can conclude after this operation, the expected values of all the real and fake nodes have the same value which is equal to $a$. Now, since all the nodes have the same expected values, the adversary gets more confused.
\end{proof}

%% file: conclusion.tex
\section{Conclusion}
\label{sec:conclusion}

The wide presence of graph data which are generated by computer systems requires graph-based privacy-preserving mechanisms. Most of the proposed privacy-preserving mechanisms ensure privacy by deleting/ adding edges. However, adding or deleting edges between two real nodes might significantly decrease the utility of the released graph data, thus, in this paper, we presented graph-based privacy preserving mechanisms which ensure privacy by adding fake nodes to the original graph. In the first part of the paper, we proposed a novel mechanism called $k-$replication method to protect the identity of users against one of the strongest attacks called seed-based attack. In the second part of the paper, we improve privacy of inter-domain routing against degree attack by adding fake nodes and edges in a way that the degree of all nodes have the same expected values. 

%% file: Allerton_Graph.bbl
\begin{thebibliography}{10}
\providecommand{\url}[1]{#1}
\csname url@samestyle\endcsname
\providecommand{\newblock}{\relax}
\providecommand{\bibinfo}[2]{#2}
\providecommand{\BIBentrySTDinterwordspacing}{\spaceskip=0pt\relax}
\providecommand{\BIBentryALTinterwordstretchfactor}{4}
\providecommand{\BIBentryALTinterwordspacing}{\spaceskip=\fontdimen2\font plus
\BIBentryALTinterwordstretchfactor\fontdimen3\font minus
  \fontdimen4\font\relax}
\providecommand{\BIBforeignlanguage}[2]{{%
\expandafter\ifx\csname l@#1\endcsname\relax
\typeout{** WARNING: IEEEtran.bst: No hyphenation pattern has been}%
\typeout{** loaded for the language `#1'. Using the pattern for}%
\typeout{** the default language instead.}%
\else
\language=\csname l@#1\endcsname
\fi
#2}}
\providecommand{\BIBdecl}{\relax}
\BIBdecl

\bibitem{Backstrom2006GroupFI}
L.~Backstrom, D.~P. Huttenlocher, J.~M. Kleinberg, and X.~Lan, ``Group
  formation in large social networks: membership, growth, and evolution,'' in
  \emph{KDD}, 2006.

\bibitem{Kossinets2008TheSO}
G.~Kossinets, J.~M. Kleinberg, and D.~J. Watts, ``The structure of information
  pathways in a social communication network,'' \emph{ArXiv}, vol.
  abs/0806.3201, 2008.

\bibitem{bergami2019social}
G.~Bergami, F.~Bertini, and D.~Montesi, ``On approximate nesting of multiple
  social network graphs: a preliminary study,'' in \emph{23rd International
  Database Engineering and Applications Symposium (IDEAS19)}.

\bibitem{kumar2010structure}
R.~Kumar, J.~Novak, and A.~Tomkins, ``Structure and evolution of online social
  networks,'' in \emph{Link mining: models, algorithms, and
  applications}.\hskip 1em plus 0.5em minus 0.4em\relax Springer, 2010, pp.
  337--357.

\bibitem{Lakshmanan2005ToDO}
L.~V.~S. Lakshmanan, R.~T. Ng, and G.~Ramesh, ``To do or not to do: The dilemma
  of disclosing anonymized data,'' in \emph{SIGMOD Conference}, 2005.

\bibitem{Ripeanu2002MappingTG}
M.~Ripeanu and I.~T. Foster, ``Mapping the gnutella network: Macroscopic
  properties of large-scale peer-to-peer systems,'' in \emph{IPTPS}, 2002.

\bibitem{asharov2017privacy}
G.~Asharov, D.~Demmler, M.~Schapira, T.~Schneider, G.~Segev, S.~Shenker, and
  M.~Zohner, ``Privacy-preserving interdomain routing at internet scale,''
  \emph{Proceedings on Privacy Enhancing Technologies}, vol. 2017, no.~3, pp.
  147--167, 2017.

\bibitem{Srivatsa2012DeanonymizingMT}
M.~Srivatsa and M.~Hicks, ``Deanonymizing mobility traces: using social network
  as a side-channel,'' in \emph{ACM Conference on Computer and Communications
  Security}, 2012.

\bibitem{wasserman1994social}
S.~Wasserman and K.~Faust, \emph{Social network analysis: Methods and
  applications}.\hskip 1em plus 0.5em minus 0.4em\relax Cambridge university
  press, 1994, vol.~8.

\bibitem{Ji2017GraphDA}
S.~Ji, P.~Mittal, and R.~A. Beyah, ``Graph data anonymization, de-anonymization
  attacks, and de-anonymizability quantification: A survey,'' \emph{IEEE
  Communications Surveys and Tutorials}, vol.~19, pp. 1305--1326, 2017.

\bibitem{narayanan2009anonymizing}
A.~Narayanan and V.~Shmatikov, ``De-anonymizing social networks,'' in
  \emph{2009 30th IEEE Symposium on Security and Privacy}.\hskip 1em plus 0.5em
  minus 0.4em\relax IEEE, 2009, pp. 173--187.

\bibitem{backstrom2007wherefore}
L.~Backstrom, C.~Dwork, and J.~Kleinberg, ``Wherefore art thou r3579x?:
  anonymized social networks, hidden patterns, and structural steganography,''
  in \emph{Proceedings of the 16th international conference on World Wide
  Web}.\hskip 1em plus 0.5em minus 0.4em\relax ACM, 2007, pp. 181--190.

\bibitem{Ji2014StructuralDD}
S.~Ji, W.~Li, M.~Srivatsa, and R.~A. Beyah, ``Structural data de-anonymization:
  Quantification, practice, and implications,'' in \emph{ACM Conference on
  Computer and Communications Security}, 2014.

\bibitem{Ji2014StructureBD}
S.~Ji, W.~Li, M.~Srivatsa, J.~He, and R.~A. Beyah, ``Structure based data
  de-anonymization of social networks and mobility traces,'' in \emph{ISC},
  2014.

\bibitem{Ji2015OnYS}
S.~Ji, W.~Li, N.~Z. Gong, P.~Mittal, and R.~A. Beyah, ``On your social network
  de-anonymizablity: Quantification and large scale evaluation with seed
  knowledge,'' in \emph{NDSS}, 2015.

\bibitem{pedarsani2011privacy}
P.~Pedarsani and M.~Grossglauser, ``On the privacy of anonymized networks,'' in
  \emph{Proceedings of the 17th ACM SIGKDD international conference on
  Knowledge discovery and data mining}.\hskip 1em plus 0.5em minus 0.4em\relax
  ACM, 2011, pp. 1235--1243.

\bibitem{yartseva2013performance}
L.~Yartseva and M.~Grossglauser, ``On the performance of percolation graph
  matching,'' in \emph{Proceedings of the first ACM conference on Online social
  networks}.\hskip 1em plus 0.5em minus 0.4em\relax ACM, 2013, pp. 119--130.

\bibitem{liu2008towards}
K.~Liu and E.~Terzi, ``Towards identity anonymization on graphs,'' in
  \emph{Proceedings of the 2008 ACM SIGMOD international conference on
  Management of data}.\hskip 1em plus 0.5em minus 0.4em\relax ACM, 2008, pp.
  93--106.

\bibitem{zhou2008preserving}
B.~Zhou and J.~Pei, ``Preserving privacy in social networks against
  neighborhood attacks.'' in \emph{ICDE}, vol.~8.\hskip 1em plus 0.5em minus
  0.4em\relax Citeseer, 2008, pp. 506--515.

\bibitem{jin2011preserving}
H.~Jin, Z.-x. ZHANG, S.-c. LIU, and S.-g. JU, ``Preserving privacy in social
  networks based on d-neighborhood subgraph anonymity,'' \emph{Application
  Research of Computers}, no.~11, p.~88, 2011.

\bibitem{wu2013d}
H.~Wu, J.~Zhang, B.~Wang, J.~Yang, and B.~Sun, ``d, k-anonymity for social
  networks publication against neighborhood attacks,'' \emph{Journal of
  Convergence Information Technology JCIT}, vol.~8, no.~2, pp. 59--67, 2013.

\bibitem{hay2008resisting}
M.~Hay, G.~Miklau, D.~Jensen, D.~Towsley, and P.~Weis, ``Resisting structural
  re-identification in anonymized social networks,'' \emph{Proceedings of the
  VLDB Endowment}, vol.~1, no.~1, pp. 102--114, 2008.

\bibitem{hay2007anonymizing}
M.~Hay, G.~Miklau, D.~Jensen, P.~Weis, and S.~Srivastava, ``Anonymizing social
  networks,'' \emph{Computer science department faculty publication series}, p.
  180, 2007.

\bibitem{newman2003structure}
M.~E. Newman, ``The structure and function of complex networks,'' \emph{SIAM
  review}, vol.~45, no.~2, pp. 167--256, 2003.

\bibitem{zou2009k}
L.~Zou, L.~Chen, and M.~T. {\"O}zsu, ``K-automorphism: A general framework for
  privacy preserving network publication,'' \emph{Proceedings of the VLDB
  Endowment}, vol.~2, no.~1, pp. 946--957, 2009.

\bibitem{lauri2016topics}
J.~Lauri and R.~Scapellato, \emph{Topics in graph automorphisms and
  reconstruction}.\hskip 1em plus 0.5em minus 0.4em\relax Cambridge University
  Press, 2016, vol. 432.

\bibitem{wang2009symmetry}
H.~Wang, G.~Yan, and Y.~Xiao, ``Symmetry in world trade network,''
  \emph{Journal of Systems Science and Complexity}, vol.~22, no.~2, pp.
  280--290, 2009.

\bibitem{ying2008randomizing}
X.~Ying and X.~Wu, ``Randomizing social networks: a spectrum preserving
  approach,'' in \emph{proceedings of the 2008 SIAM International Conference on
  Data Mining}.\hskip 1em plus 0.5em minus 0.4em\relax SIAM, 2008, pp.
  739--750.

\bibitem{machanavajjhala2006diversity}
A.~Machanavajjhala, J.~Gehrke, D.~Kifer, and M.~Venkitasubramaniam,
  ``l-diversity: Privacy beyond k-anonymity,'' in \emph{22nd International
  Conference on Data Engineering (ICDE'06)}.\hskip 1em plus 0.5em minus
  0.4em\relax IEEE, 2006, pp. 24--24.

\bibitem{wong2006alpha}
R.~C.-W. Wong, J.~Li, A.~W.-C. Fu, and K.~Wang, ``($\alpha$, k)-anonymity: an
  enhanced k-anonymity model for privacy preserving data publishing,'' in
  \emph{Proceedings of the 12th ACM SIGKDD international conference on
  Knowledge discovery and data mining}.\hskip 1em plus 0.5em minus 0.4em\relax
  ACM, 2006, pp. 754--759.

\bibitem{cheng2010k}
J.~Cheng, A.~W.-c. Fu, and J.~Liu, ``K-isomorphism: privacy preserving network
  publication against structural attacks,'' in \emph{Proceedings of the 2010
  ACM SIGMOD International Conference on Management of data}.\hskip 1em plus
  0.5em minus 0.4em\relax ACM, 2010, pp. 459--470.

\bibitem{yang2014secure}
J.~Yang, B.~Wang, X.~Yang, H.~Zhang, and G.~Xiang, ``A secure k-automorphism
  privacy preserving approach with high data utility in social networks,''
  \emph{Security and Communication Networks}, vol.~7, no.~9, pp. 1399--1411,
  2014.

\bibitem{onaran2016optimal}
E.~Onaran, S.~Garg, and E.~Erkip, ``Optimal de-anonymization in random graphs
  with community structure,'' in \emph{2016 50th Asilomar Conference on
  Signals, Systems and Computers}.\hskip 1em plus 0.5em minus 0.4em\relax IEEE,
  2016, pp. 709--713.

\bibitem{kazemi2015growing}
E.~Kazemi, S.~H. Hassani, and M.~Grossglauser, ``Growing a graph matching from
  a handful of seeds,'' \emph{Proceedings of the VLDB Endowment}, vol.~8,
  no.~10, pp. 1010--1021, 2015.

\bibitem{lyzinski2014seeded}
V.~Lyzinski, D.~E. Fishkind, and C.~E. Priebe, ``Seeded graph matching for
  correlated erd{\"o}s-r{\'e}nyi graphs.'' \emph{Journal of Machine Learning
  Research}, vol.~15, no.~1, pp. 3513--3540, 2014.

\bibitem{kazemi2015can}
E.~Kazemi, L.~Yartseva, and M.~Grossglauser, ``When can two unlabeled networks
  be aligned under partial overlap?'' in \emph{2015 53rd Annual Allerton
  Conference on Communication, Control, and Computing (Allerton)}.\hskip 1em
  plus 0.5em minus 0.4em\relax IEEE, 2015, pp. 33--42.

\bibitem{cullina2016improved}
D.~Cullina and N.~Kiyavash, ``Improved achievability and converse bounds for
  erdos-r{\'e}nyi graph matching,'' in \emph{ACM SIGMETRICS Performance
  Evaluation Review}, vol.~44, no.~1.\hskip 1em plus 0.5em minus 0.4em\relax
  ACM, 2016, pp. 63--72.

\bibitem{Korula2014AnER}
N.~Korula and S.~Lattanzi, ``An efficient reconciliation algorithm for social
  networks,'' \emph{ArXiv}, vol. abs/1307.1690, 2014.

\bibitem{shirani2017seeded}
F.~Shirani, S.~Garg, and E.~Erkip, ``Seeded graph matching: Efficient
  algorithms and theoretical guarantees,'' in \emph{2017 51st Asilomar
  Conference on Signals, Systems, and Computers}.\hskip 1em plus 0.5em minus
  0.4em\relax IEEE, 2017, pp. 253--257.

\bibitem{kazemi2016network}
E.~Kazemi, ``Network alignment: Theory, algorithms, and applications,'' EPFL,
  Tech. Rep., 2016.

\bibitem{Ji2016SeedBasedDQ}
S.~Ji, W.~Li, N.~Z. Gong, P.~Mittal, and R.~A. Beyah, ``Seed-based
  de-anonymizability quantification of social networks,'' \emph{IEEE
  Transactions on Information Forensics and Security}, vol.~11, pp. 1398--1411,
  2016.

\bibitem{974527}
L.~Gao, ``On inferring autonomous system relationships in the internet,''
  \emph{IEEE/ACM Transactions on Networking}, vol.~9, no.~6, pp. 733--745, Dec
  2001.

\bibitem{TimGriffin_SIGCOMM99}
\BIBentryALTinterwordspacing
T.~G. Griffin and G.~Wilfong, ``An analysis of bgp convergence properties,'' in
  \emph{Proceedings of the Conference on Applications, Technologies,
  Architectures, and Protocols for Computer Communication}, ser. SIGCOMM
  '99.\hskip 1em plus 0.5em minus 0.4em\relax New York, NY, USA: ACM, 1999, pp.
  277--288. [Online]. Available: \url{http://doi.acm.org/10.1145/316188.316231}
\BIBentrySTDinterwordspacing

\bibitem{Sobrinho:2003:NRP:863955.863963}
\BIBentryALTinterwordspacing
J.~a.~L. Sobrinho, ``Network routing with path vector protocols: Theory and
  applications,'' in \emph{Proceedings of the 2003 Conference on Applications,
  Technologies, Architectures, and Protocols for Computer Communications}, ser.
  SIGCOMM '03.\hskip 1em plus 0.5em minus 0.4em\relax New York, NY, USA: ACM,
  2003, pp. 49--60. [Online]. Available:
  \url{http://doi.acm.org/10.1145/863955.863963}
\BIBentrySTDinterwordspacing

\bibitem{Sosnovich:2015:AIR:3089605.3089614}
\BIBentryALTinterwordspacing
A.~Sosnovich, O.~Grumberg, and G.~Nakibly, ``Analyzing internet routing
  security using model checking,'' in \emph{Proceedings of the 20th
  International Conference on Logic for Programming, Artificial Intelligence,
  and Reasoning - Volume 9450}, ser. LPAR-20 2015.\hskip 1em plus 0.5em minus
  0.4em\relax New York, NY, USA: Springer-Verlag New York, Inc., 2015, pp.
  112--129. [Online]. Available:
  \url{https://doi.org/10.1007/978-3-662-48899-7_9}
\BIBentrySTDinterwordspacing

\bibitem{complexrelationship}
V.~Giotsas, M.~Luckie, B.~Huffaker, and k.~claffy, ``{Inferring Complex AS
  Relationships},'' in \emph{Internet Measurement Conference (IMC)}, Nov 2014,
  pp. 23--30.

\bibitem{Wang:2009:NBM:1555349.1555375}
\BIBentryALTinterwordspacing
Y.~Wang, M.~Schapira, and J.~Rexford, ``Neighbor-specific bgp: More flexible
  routing policies while improving global stability,'' in \emph{Proceedings of
  the Eleventh International Joint Conference on Measurement and Modeling of
  Computer Systems}, ser. SIGMETRICS '09.\hskip 1em plus 0.5em minus
  0.4em\relax New York, NY, USA: ACM, 2009, pp. 217--228. [Online]. Available:
  \url{http://doi.acm.org/10.1145/1555349.1555375}
\BIBentrySTDinterwordspacing

\bibitem{7218436}
R.~{Klöti}, V.~{Kotronis}, B.~{Ager}, and X.~{Dimitropoulos},
  ``Policy-compliant path diversity and bisection bandwidth,'' in \emph{2015
  IEEE Conference on Computer Communications (INFOCOM)}, April 2015, pp.
  675--683.

\bibitem{de2003large}
M.~De~Santo, P.~Foggia, C.~Sansone, and M.~Vento, ``A large database of graphs
  and its use for benchmarking graph isomorphism algorithms,'' \emph{Pattern
  Recognition Letters}, vol.~24, no.~8, pp. 1067--1079, 2003.

\bibitem{beineke2004topics}
L.~W. Beineke, R.~J. Wilson, P.~J. Cameron \emph{et~al.}, \emph{Topics in
  algebraic graph theory}.\hskip 1em plus 0.5em minus 0.4em\relax Cambridge
  University Press, 2004, vol. 102.

\bibitem{Nazanin_IT}
N.~Takbiri, A.~Houmansadr, D.~L. Goeckel, and H.~Pishro{-}Nik, ``Matching
  anonymized and obfuscated time series to users' profiles,'' \emph{IEEE
  Transactions on Information Theory}, vol.~65, no.~2, pp. 724--741, 2019.

\bibitem{tifs2016}
Z.~Montazeri, A.~Houmansadr, and H.~Pishro-Nik, ``{Achieving Perfect Location
  Privacy in Wireless Devices Using Anonymization},'' \emph{IEEE Transaction on
  Information Forensics and Security}, vol.~12, no.~11, pp. 2683--2698, 2017.

\bibitem{diaz2002towards}
C.~Diaz, S.~Seys, J.~Claessens, and B.~Preneel, ``Towards measuring
  anonymity,'' in \emph{International Workshop on Privacy Enhancing
  Technologies}.\hskip 1em plus 0.5em minus 0.4em\relax Springer, 2002, pp.
  54--68.

\bibitem{serjantov2002towards}
A.~Serjantov and G.~Danezis, ``Towards an information theoretic metric for
  anonymity,'' in \emph{International Workshop on Privacy Enhancing
  Technologies}.\hskip 1em plus 0.5em minus 0.4em\relax Springer, 2002, pp.
  41--53.

\bibitem{nilizadeh2014community}
S.~Nilizadeh, A.~Kapadia, and Y.-Y. Ahn, ``Community-enhanced de-anonymization
  of online social networks,'' in \emph{Proceedings of the 2014 acm sigsac
  conference on computer and communications security}.\hskip 1em plus 0.5em
  minus 0.4em\relax ACM, 2014, pp. 537--548.

\bibitem{ji2018quantifying}
S.~Ji, T.~Du, Z.~Hong, T.~Wang, and R.~Beyah, ``Quantifying graph anonymity,
  utility, and de-anonymity,'' in \emph{IEEE INFOCOM 2018-IEEE Conference on
  Computer Communications}.\hskip 1em plus 0.5em minus 0.4em\relax IEEE, 2018,
  pp. 1736--1744.

\bibitem{ji2016evaluating}
S.~Ji, ``Evaluating the security of anonymized big graph/structural data,''
  Ph.D. dissertation, Georgia Institute of Technology, 2016.

\bibitem{gross2005graph}
J.~L. Gross and J.~Yellen, \emph{Graph theory and its applications}.\hskip 1em
  plus 0.5em minus 0.4em\relax CRC press, 2005.

\bibitem{routeview}
\BIBentryALTinterwordspacing
U.~of~Oregon, ``Route views.'' [Online]. Available: \url{www.routeviews.org}
\BIBentrySTDinterwordspacing

\bibitem{ripe}
\BIBentryALTinterwordspacing
R.~NCC, ``Routing information service.'' [Online]. Available:
  \url{https://www.ripe.net/analyse/internet-measurements/routing-information-service-ris}
\BIBentrySTDinterwordspacing

\bibitem{hay2009accurate}
M.~Hay, C.~Li, G.~Miklau, and D.~Jensen, ``Accurate estimation of the degree
  distribution of private networks,'' in \emph{2009 Ninth IEEE International
  Conference on Data Mining}.\hskip 1em plus 0.5em minus 0.4em\relax IEEE,
  2009, pp. 169--178.

\end{thebibliography}
